\documentclass[11pt,letterpaper]{article}
\usepackage{comment}
\usepackage{xspace}
\usepackage[T1]{fontenc}
\usepackage[margin=1in]{geometry}
\usepackage{mathtools,amssymb,amsthm}
\usepackage{newtxtext,newtxmath}
\usepackage{microtype}
\usepackage[dvipsnames]{xcolor}
\usepackage{tikz}
\usetikzlibrary{arrows.meta,calc,positioning,decorations.pathreplacing}
\usepackage{float}
\usepackage{wrapfig}
\usepackage{booktabs}
\usepackage{enumitem}
\usepackage{url}
\usepackage{aliascnt}
\usepackage[hidelinks]{hyperref}
\usepackage[capitalise,noabbrev]{cleveref}
\input{glyphtounicode}
\hypersetup{
  pdftitle={Dense Pinwheel Packing Is Strongly NP-Complete},
  pdfauthor={},
  pdfsubject={Computational complexity of pinwheel packing},
  pdfkeywords={pinwheel scheduling, exact covering systems, NP-completeness},
  pdflang={en-US}
}

\setlist{nosep,leftmargin=2.1em}

\newtheorem{theorem}{Theorem}[section]
\newaliascnt{lemma}{theorem}
\newtheorem{lemma}[lemma]{Lemma}
\aliascntresetthe{lemma}
\newaliascnt{proposition}{theorem}
\newtheorem{proposition}[proposition]{Proposition}
\aliascntresetthe{proposition}
\newaliascnt{corollary}{theorem}

\aliascntresetthe{corollary}
\theoremstyle{definition}
\newaliascnt{definition}{theorem}
\newtheorem{definition}[definition]{Definition}
\aliascntresetthe{definition}
\newaliascnt{remark}{theorem}

\aliascntresetthe{remark}

\newcommand{\Z}{\mathbb Z}
\newcommand{\PP}{\textsc{Pinwheel Packing}\xspace}
\newcommand{\DPP}{\textsc{Dense Pinwheel Packing}\xspace}
\newcommand{\TP}{\textsc{Sparse Tripartite Triangle Partition}\xspace}
\newcommand{\dens}{\operatorname{dens}}

\definecolor{pwblue}{RGB}{73,126,169}
\definecolor{pwbluefill}{RGB}{219,234,245}
\definecolor{pworange}{RGB}{201,119,48}
\definecolor{pworangefill}{RGB}{248,226,199}
\definecolor{pwgreen}{RGB}{77,139,105}
\definecolor{pwgreenfill}{RGB}{218,237,225}
\definecolor{pwred}{RGB}{173,73,73}
\definecolor{pwgray}{RGB}{105,112,119}
\definecolor{pwgrayfill}{RGB}{238,240,242}

\tikzset{
  pw arrow/.style={-{Latex[length=2mm]},draw=pwgray,semithick},
  pw brace/.style={decorate,decoration={brace,amplitude=3pt}},
  pw dot/.style={circle,draw=pwgray,fill=white,inner sep=0pt,minimum size=4.5mm},
  pw task/.style={
    rounded corners=1.5pt,
    draw=pwgray,
    fill=pwgrayfill,
    inner sep=2pt,
    minimum height=5.5mm
  },
  pw x/.style={circle,draw=pwblue,fill=pwbluefill,inner sep=0pt,minimum size=5mm},
  pw y/.style={circle,draw=pworange,fill=pworangefill,inner sep=0pt,minimum size=5mm},
  pw z/.style={circle,draw=pwgreen,fill=pwgreenfill,inner sep=0pt,minimum size=5mm}
}

\newcommand{\pwMarkerCollisionDiagram}{%
  \begin{tikzpicture}[font=\scriptsize,x=0.54cm,y=0.48cm]
    \def\pwGridMax{4}
    \def\pwBlueIndex{1}
    \def\pwOrangeIndex{3}
    \coordinate (leftGrid) at (0,0);
    \coordinate (rightGrid) at ($(leftGrid)+(6,0)$);
    \coordinate (leftBlue) at ($(leftGrid)+(\pwBlueIndex,0)$);
    \coordinate (leftOrange) at ($(leftGrid)+(0,\pwOrangeIndex)$);
    \coordinate (leftHit) at
      ($(leftGrid)+(\pwBlueIndex,\pwOrangeIndex)$);
    \coordinate (rightBlue) at
      ($(rightGrid)+(\pwBlueIndex,\pwBlueIndex)$);
    \coordinate (rightOrange) at
      ($(rightGrid)+(\pwOrangeIndex,\pwOrangeIndex)$);
    \coordinate (rightMiss) at
      ($(rightGrid)+(\pwBlueIndex,\pwOrangeIndex)$);

    \node[text=pwgray] at ($(leftGrid)+(2,5)$) {\(\gcd(m_u,m_v)=1\)};
    \draw[pwblue,line width=2pt,opacity=0.55]
      ($(leftBlue)+(0,-0.25)$) -- ($(leftBlue)+(0,4.25)$);
    \draw[pworange,line width=2pt,opacity=0.55]
      ($(leftOrange)+(-0.25,0)$) -- ($(leftOrange)+(4.25,0)$);
    \foreach \xx in {0,...,4} {
      \foreach \yy in {0,...,4} {
        \fill[pwgray] ($(leftGrid)+(\xx,\yy)$) circle (0.65pt);
      }
    }
    \fill[pwred] (leftHit) circle (2.7pt);
    \node[text=pwblue] at ($(leftBlue)+(0,-0.65)$) {\(W_u\)};
    \node[text=pworange] at ($(leftOrange)+(-0.62,0)$) {\(W_v\)};

    \draw[pwgray,dashed] ($(rightGrid)+(-1,-0.8)$)
      -- ($(rightGrid)+(-1,5.35)$);

    \node[text=pwgray] at ($(rightGrid)+(2,5)$)
      {\(p_T\mid\gcd(m_u,m_v)\)};
    \draw[pwblue,line width=2pt,opacity=0.55]
      ($(rightGrid)+(\pwBlueIndex,-0.25)$)
      -- ($(rightGrid)+(\pwBlueIndex,4.25)$);
    \draw[pworange,line width=2pt,opacity=0.55]
      ($(rightGrid)+(-0.25,\pwOrangeIndex)$)
      -- ($(rightGrid)+(4.25,\pwOrangeIndex)$);
    \draw[pwgray!45] (rightGrid) -- ($(rightGrid)+(\pwGridMax,\pwGridMax)$);
    \foreach \xx in {0,...,\pwGridMax} {
      \foreach \yy in {0,...,\pwGridMax} {
        \draw[pwgray!45] ($(rightGrid)+(\xx,\yy)$) circle (0.65pt);
      }
      \fill[pwgray] ($(rightGrid)+(\xx,\xx)$) circle (1.25pt);
    }
    \fill[pwblue] (rightBlue) circle (2.5pt);
    \fill[pworange] (rightOrange) circle (2.5pt);
    \draw[pwred,semithick] ($(rightMiss)+(-0.18,0.18)$)
      -- ($(rightMiss)+(0.18,-0.18)$);
    \draw[pwred,semithick] ($(rightMiss)+(-0.18,-0.18)$)
      -- ($(rightMiss)+(0.18,0.18)$);
    \node[text=pwblue] at
      ($(rightGrid)+(\pwBlueIndex,-0.65)$) {\(W_u\)};
    \node[text=pworange] at
      ($(rightGrid)+(4.62,\pwOrangeIndex)$) {\(W_v\)};
  \end{tikzpicture}%
}

\newcommand{\pwLocalQuotientDiagram}{%
  \begin{tikzpicture}[x=0.88cm,y=0.68cm,font=\normalsize]
    \definecolor{xblue}{RGB}{203,224,242}
    \definecolor{yorange}{RGB}{250,220,180}
    \definecolor{zgreen}{RGB}{205,232,205}
    \def\pwSpecialColumns{3}
    \def\pwTotalColumns{7}
    \def\pwRows{3}
    \coordinate (grid) at (0,0);
    \coordinate (remaining) at ($(grid)+(\pwSpecialColumns,0)$);
    \coordinate (gridTopRight) at
      ($(grid)+(\pwTotalColumns,\pwRows)$);

    \foreach \column/\tone/\name in {0/xblue/x,1/yorange/y,2/zgreen/z} {
      \fill[\tone] ($(grid)+(\column,0)$) rectangle +(1,\pwRows);
      \foreach \row in {0,...,2} {
        \node at ($(grid)+(\column+0.5,\row+0.5)$) {\(\name\)};
      }
    }
    \foreach \row/\tone/\name in {0/zgreen/z,1/yorange/y,2/xblue/x} {
      \fill[\tone] ($(remaining)+(0,\row)$)
        rectangle ($(grid)+(\pwTotalColumns,\row+1)$);
      \foreach \column in {3,...,6} {
        \node at ($(grid)+(\column+0.5,\row+0.5)$) {\(\name\)};
      }
    }
    \draw[step=1] (grid) grid (gridTopRight);
    \foreach \column/\name in {0/x,1/y,2/z} {
      \node[below] at ($(grid)+(\column+0.5,0)$) {\(c_\name\)};
    }
    \draw[decorate,decoration={brace,amplitude=4pt,mirror}]
      ($(remaining)+(0,-0.42)$) -- ($(grid)+(\pwTotalColumns,-0.42)$)
      node[midway,below=5pt] {\(p-3\) remaining columns};
    \foreach \row/\label in {0/2,1/1,2/0} {
      \node[left] at ($(grid)+(0,\row+0.5)$) {\(\label\)};
    }
  \end{tikzpicture}%
}
\newcommand{\pwDiamondDiagram}{%
  \begin{tikzpicture}[font=\scriptsize,x=0.9cm,y=0.9cm]
    \def\pwHalfWidth{1.6}
    \def\pwHalfHeight{1.15}
    \coordinate (diamondCenter) at (1.6,0);

    \node[pw x] (e) at
      ($(diamondCenter)+(-\pwHalfWidth,0)$) {\(e\)};
    \node[pw x] (a) at
      ($(diamondCenter)+(\pwHalfWidth,0)$) {\(a_{\tau,e}\)};
    \node[pw y] (u) at
      ($(diamondCenter)+(0,\pwHalfHeight)$) {\(u_{\tau,e}\)};
    \node[pw z] (v) at
      ($(diamondCenter)+(0,-\pwHalfHeight)$) {\(v_{\tau,e}\)};

    \draw[pwgray] (e) -- (u) -- (a) -- (v) -- (e);
    \draw[pwgray] (u) -- (v);
    \draw[pwgray,dashed] (e) -- (a);

    \node[above=3pt of u,text=pwgray] {middle vertices};
    \node[below=5pt of v,text=pwgray]
      {\(\{e,u_{\tau,e},v_{\tau,e}\}
        \quad
        \{a_{\tau,e},u_{\tau,e},v_{\tau,e}\}\)};
  \end{tikzpicture}%
}

\newcommand{\pwTripleGadgetDiagram}{%
  \begin{tikzpicture}[font=\scriptsize,scale=0.57]
    \def\pwApexHalfWidth{1.65}
    \def\pwApexTop{1.8}
    \def\pwApexBottom{-1.05}
    \def\pwElementHalfWidth{3.75}
    \def\pwElementTop{4.25}
    \def\pwElementBottom{-2.25}
    \def\pwDiamondHalfWidth{0.62cm}
    \pgfmathsetmacro{\pwCentralY}{(\pwApexTop+2*\pwApexBottom)/3}

    \coordinate (central) at (0,\pwCentralY);
    \node[pw x] (ar) at (0,\pwApexTop) {\(a_r\)};
    \node[pw y] (ab) at (-\pwApexHalfWidth,\pwApexBottom) {\(a_b\)};
    \node[pw z] (ay) at (\pwApexHalfWidth,\pwApexBottom) {\(a_y\)};
    \draw[pwgray] (ar) -- (ab)
      (ab) -- (ay)
      (ay) -- (ar);

    \node[pw x] (er) at (0,\pwElementTop) {\(r\)};
    \node[pw y] (eb) at (-\pwElementHalfWidth,\pwElementBottom) {\(b\)};
    \node[pw z] (ey) at (\pwElementHalfWidth,\pwElementBottom) {\(y\)};

    \coordinate (rmid) at ($(ar)!0.5!(er)$);
    \node[pw y] (rleft) at ($(rmid)!\pwDiamondHalfWidth!90:(er)$) {};
    \node[pw z] (rright) at ($(rmid)!\pwDiamondHalfWidth!-90:(er)$) {};
    \draw[pwgray] (er) -- (rleft) -- (ar) -- (rright) -- (er);
    \draw[pwgray] (rleft) -- (rright);

    \coordinate (bmid) at ($(ab)!0.5!(eb)$);
    \node[pw x] (bup) at ($(bmid)!\pwDiamondHalfWidth!90:(eb)$) {};
    \node[pw z] (bdown) at ($(bmid)!\pwDiamondHalfWidth!-90:(eb)$) {};
    \draw[pwgray] (eb) -- (bup) -- (ab) -- (bdown) -- (eb);
    \draw[pwgray] (bup) -- (bdown);

    \coordinate (ymid) at ($(ay)!0.5!(ey)$);
    \node[pw x] (yup) at ($(ymid)!\pwDiamondHalfWidth!90:(ey)$) {};
    \node[pw y] (ydown) at ($(ymid)!\pwDiamondHalfWidth!-90:(ey)$) {};
    \draw[pwgray] (ey) -- (yup) -- (ay) -- (ydown) -- (ey);
    \draw[pwgray] (yup) -- (ydown);

    \coordinate (elementside) at ($(er)!0.67!(rmid)$);
    \coordinate (apexside) at ($(ar)!0.67!(rmid)$);

    \node[align=center,text=pwgray] (elementlabel) at ($(er)+(-2.3,0.3)$)
      {element\\vertex};
    \draw[pwgray] (elementlabel.east) -- (er);
    \node[align=center,text=pwgray] (middlelabel) at ($(rmid)+(-2.8,0.295)$)
      {middle\\vertex};
    \draw[pwgray] (middlelabel.east) -- (rleft);
    \node[align=center,text=pwgray] (apexlabel) at ($(ar)+(-2.25,0.05)$)
      {apex\\vertex};
    \draw[pwgray] (apexlabel.east) -- (ar);

    \node[align=center,text=pwgray] (elementtrilabel) at
      ($(elementside)+(2.35,0.29)$)
      {element-side\\triangle};
    \draw[pwgray] (elementtrilabel.west) -- (elementside);
    \node[align=center,text=pwgray] (apextrilabel) at
      ($(apexside)+(2.35,-0.14)$)
      {apex-side\\triangle};
    \draw[pwgray] (apextrilabel.west) -- (apexside);
    \node[align=center,text=pwgray] (centrallabel) at
      ($(central)+(3.05,0.45)$)
      {central\\triangle};
    \draw[pwgray] (centrallabel.west) -- (central);
  \end{tikzpicture}%
}

\title{Dense Pinwheel Packing Is Strongly NP-Complete}

\author{
  Yusuke Kobayashi%
  \thanks{Research Institute for Mathematical Sciences, Kyoto University,
  Kyoto, Japan. Email: \texttt{yusuke@kurims.kyoto-u.ac.jp}}
  \and
  Bingkai Lin%
  \thanks{State Key Laboratory for Novel Software Technology,
  Nanjing University, Nanjing, China.
  Email: \texttt{lin@nju.edu.cn}}
  \and
  Joseph Swernofsky%
  \thanks{Department of Electrical Engineering,
  National Taiwan University, Taipei, Taiwan.
  Email: \texttt{d14921b19@ntu.edu.tw}}
}

\date{}

\begin{document}
\maketitle

\begin{abstract}
  An instance of {\sc Pinwheel Packing} is a list of positive integers
  $a_1,\ldots,a_k$.  A feasible schedule assigns one task to every integer
  time so that every interval of $a_i$ consecutive times contains task $i$.
  The instance is \emph{dense} when $\sum_i1/a_i=1$.

  We prove that {\sc Dense Pinwheel Packing} is NP-complete even when every period
  is encoded in unary and equal periods are listed as distinct tasks.
  Consequently, the usual binary-encoded problem is strongly NP-complete.
  Kleinberg and Mishra also prove NP-completeness
  \cite[Corollary~5.1]{KleinbergMishra2026}, but their reduction uses periods
  of exponential numerical size and therefore yields only weak NP-hardness.
  Our proof uses a direct reduction from triangle partition in a sparse
  tripartite graph.  If each of the three parts of the source graph has $n$
  vertices, the reduction produces $O(n^4\log^3 n)$ explicitly listed tasks,
  each with period $O(n^4\log^3 n)$; consequently, its full unary encoding
  has length $O(n^8\log^6 n)$.
\end{abstract}


\section{Introduction}
The \PP
problem was introduced in~\cite{HolteEtAl1989Pinwheel}.  Given recurrent tasks with periods $a_1,\ldots,a_k$, the problem
asks whether the tasks can be scheduled on a discrete timeline, with at most
one task executed in each slot, so that every interval of $a_i$ consecutive
slots contains an execution of task $i$.  These recurrence requirements imply
the necessary density bound
\[
  \sum_{i=1}^k \frac{1}{a_i}\le 1.
\]
Instances attaining equality are called \emph{dense}.  At density one,
feasibility is particularly rigid: it is equivalent to choosing one
congruence class modulo $a_i$ for each task so that the chosen classes are
pairwise disjoint, in which case they necessarily partition the timeline.

The representation of the period list is important for the complexity of the
problem.  We use the explicit-list formulation: entries at different
positions represent distinct tasks even when their periods are equal, and
repeated periods are not encoded by compressed multiplicities.  Unary and
binary encoding refer only to the representation of the individual periods.

Mok, Rosier, Tulchinsky, and Varvel proved that an exact variant of \PP is NP-complete even when the
periods are encoded in unary \cite[Theorem~1]{MokEtAl1989}.  
In this variant, each task $i$ is executed exactly once in every interval of $a_i$ consecutive slots. 
This result does not settle the dense problem, because exact packing instances may not be dense.  At density one, however, the packing and covering
formulations, together with their exact variants, all coincide
\cite[Sections~4--5]{KawSurvey}.  In this setting, Kawamura and Soejima
formulated the equivalent Disjoint Covering System Problem and conjectured
that it is NP-complete even when the moduli are encoded in unary
\cite[Conjecture~19]{KawamuraSoejima2020}.

For the explicit-list problem with binary-encoded periods, Jacobs and Longo
obtained an earlier conditional lower bound \cite{JL14}.  This was
subsequently strengthened by Kobayashi, Lin, and Swernofsky, who proved that
a polynomial-time algorithm for \DPP would imply
\[
  \mathrm{NP}\subseteq
  \mathrm{DTIME}\!\left(N^{O(\log\log N)}\right),
\]
where $N$ denotes the input length
\cite[Theorem~1.2]{KobayashiLinSwernofsky2026}.  This implication provides
strong conditional evidence against a polynomial-time algorithm, but it is
not itself an NP-hardness result.

Kleinberg and Mishra recently proved NP-hardness of \DPP{}
\cite[Corollary~5.1]{KleinbergMishra2026}.  Their construction takes
cumulative products of a linear number of polynomially bounded factors,
producing periods with polynomial-length binary encodings but exponentially
large numerical values.  Consequently, the reduction has polynomial output
length under binary encoding but not under unary encoding.  It therefore
establishes only weak NP-hardness and does not prove hardness when the periods
are unary-encoded or polynomially bounded
\cite[discussion following Lemma~4.14]{KleinbergMishra2026}.

In contrast, from a source graph with $n$ vertices in each tripartition
class, this paper gives a reduction that produces $O(n^4\log^3 n)$ explicitly listed tasks, each
with period $O(n^4\log^3 n)$.  The complete unary encoding of the resulting
instance has length $O(n^8\log^6 n)$.  Hence the reduction remains
polynomial-time under unary encoding and establishes strong NP-hardness.

\begin{theorem}[Main theorem]\label{thm:main}
  \DPP{} is NP-complete even when all explicitly listed periods are encoded
  in unary.  Consequently, its usual binary-encoded version is strongly
  NP-complete.
\end{theorem}

The reduction starts from a tripartite graph with $n$ vertices in
each part and at most three triangles through each vertex.  We give each graph
triangle $T$ a distinct prime $p_T$, called its \emph{marker}, and let $m_v$
be the product of the markers of the triangles containing vertex $v$.  For each
vertex $v$, the target instance contains one \emph{witness} task of period
$nm_v$ and $m_v-3$ \emph{cell} tasks of period $3nm_v$.

The residue of a witness class modulo $n$ is the vertex's \emph{color}.
Witnesses force the three vertices of each color to form a triangle.
The cell tasks fill the remaining residue classes by a direct local
construction.  Because each marker product has at most three prime factors,
all task multiplicities and periods are polynomial in the source size.  We
state the restricted source result in \cref{thm:source} and prove it in
\cref{app:source}.

\paragraph{Organization.}
\Cref{sec:model} defines \DPP{} and proves residue-class rigidity.
\Cref{sec:source-reduction} defines \TP{} and our reduction.
\Cref{sec:soundness,sec:completeness} prove that this reduction is sound and complete.
\cref{sec:size} proves that the resulting instance is polynomial in size.  Finally, in \cref{app:source} we prove that \TP{} is NP-complete.

\section{The language and dense rigidity}\label{sec:model}


\subsection{Instances and feasibility}
An instance of \PP is a nonempty sequence $A=(a_1,\ldots,a_k)$ of positive integers.
Different list positions are distinct tasks, even when their periods are
equal; no multiplicity notation is allowed.  We use the unary encoding
\[
  \langle A\rangle_{\mathrm{un}}
    =1^{a_1}0\,1^{a_2}0\cdots 1^{a_k}0,
\]
whose length is $k+\sum_i a_i$.  The usual binary representation encodes the
same explicit list with the periods written in binary.

\begin{definition}[Pinwheel packing]
  A packing of $A$ is a map $\sigma:\Z\to[k]$.  It is \emph{feasible} if, for every
  task $i$ and every $s\in\Z$, there exists
  $t\in\{s,s+1,\ldots,s+a_i-1\}$ such that $\sigma(t)=i$.
\end{definition}

Thus every slot contains exactly one task.  Allowing idle slots does not
change feasibility, because filling an idle slot by any task cannot destroy a
window condition.  Define
\[
  \dens(A)=\sum_{i=1}^k\frac1{a_i}.
\]
The decision problem \DPP{} asks whether $\dens(A)=1$ and $A$ admits a
feasible packing.

The density test is exact and polynomial-time.  If $P=\prod_i a_i$, then
$\dens(A)=1$ if and only if $\sum_i P/a_i=P$.  Moreover, $P$ has
$O(\sum_i\log(a_i+1))$ bits, which is polynomial in the unary input length.

\subsection{Rigidity at density one}

For integers $r$ and $a>0$, write $[r]_a=r+a\Z$.

\begin{lemma}[Dense rigidity]\label{lem:rigidity}
  Let $A=(a_1,\ldots,a_k)$ satisfy $\dens(A)=1$.  Then $A$
  is feasible if and only if there are residues $0\le r_i<a_i$ such that the
  classes $[r_i]_{a_i}$ are pairwise disjoint.  In that event the classes
  automatically partition $\Z$.
\end{lemma}

This is the standard density-one equivalence between pinwheel packing and its
exact residue-class formulation; see \cite[Section~5]{KawSurvey}.

We will repeatedly use the generalized Chinese remainder theorem in the
following form:
\begin{equation}\label{eq:crt-intersection}
  [r]_a\cap[s]_b\ne\varnothing
  \quad\Longleftrightarrow\quad
  r\equiv s\pmod{\gcd(a,b)}.
\end{equation}

The residue-class characterization also shows that \DPP{} belongs to
$\mathrm{NP}$ under either representation: guess the residues $r_i$ in
binary, check the density, and use \cref{eq:crt-intersection} to test pairwise
disjointness.

\section{Source problem and reduction}\label{sec:source-reduction}

\subsection{The source problem}

\begin{definition}
An instance of \TP{} consists of a tripartite graph
\[
  G=(X\mathbin{\dot\cup}Y\mathbin{\dot\cup}Z,E),
  \qquad |X|=|Y|=|Z|=n\ge1,
\]
where $E$ is an undirected edge set satisfying
\[
  E\subseteq
  \bigl\{\{x,y\}:x\in X,\ y\in Y\bigr\}
  \cup
  \bigl\{\{x,z\}:x\in X,\ z\in Z\bigr\}
  \cup
  \bigl\{\{y,z\}:y\in Y,\ z\in Z\bigr\}.
\]
A triangle is a triple
$\{x,y,z\} \subseteq X\cup Y\cup Z$ such that
$\{x,y\},\{x,z\},\{y,z\}\in E$.
We require $G$ to have maximum degree at most six and require each vertex
to belong to at least one and at most three triangles.  The graph $G$ is in
\TP{} if and only if its $3n$ vertices can be partitioned into $n$ triangles.
\end{definition}

Counting vertex--triangle incidences shows that every valid source graph has
at most $3n$ triangles: its $3n$ vertices supply at most $9n$
incidences, and every triangle uses three.

\begin{theorem}[Restricted triangle partition]\label{thm:source}
  \TP{} is NP-complete under polynomial-time many-one reductions.
\end{theorem}

We give a self-contained reduction from bounded-occurrence 3DM in \cref{app:source}.

\subsection{The reduction}

Given a source instance $G$ of \TP, enumerate the set $\mathcal T$ of all triangles
and assign successive primes $p_T>\max\{3,n\}$.
Call $p_T$ the \emph{marker} of $T$.
For each vertex $v$, define its \emph{marker product}
\begin{equation}\label{eq:marker-product}
  m_v := \prod_{T\in\mathcal T:\,v\in T}p_T.
\end{equation}
Let $A(G)$ be the explicit target instance of \DPP containing, for every vertex $v$, the
following tasks:
\begin{equation}\label{eq:block}
  \underbrace{1\text{ task of period }nm_v}_{\text{witness of }v}
  \quad\text{and}\quad
  \underbrace{m_v-3\text{ tasks of period }3nm_v}_{\text{cells of }v}.
\end{equation}
Each $m_v$ contains at least one marker prime.  Every marker
prime exceeds three, so $m_v\ge5$ and hence $m_v-3\ge2$.

Formally, $A(G)$ is only an indexed list of integer periods.  The terms
\emph{witness} and \emph{cell} are names used in the proof; they are not
additional task types in the target instance.

\begin{proposition}[Exact density]\label{prop:density}
  The instance $A(G)$ has density exactly one.
\end{proposition}

\begin{proof}
  The reciprocal weight of the block for $v$ is
  \[
    \frac1{nm_v}+\frac{m_v-3}{3nm_v}=\frac1{3n}.
  \]
  There are $3n$ vertex blocks, so the total density is exactly one.
\end{proof}

\section{Soundness}\label{sec:soundness}

Fix a source graph $G$, and let $A(G)$ be the explicit target instance from
\cref{eq:block}.  Suppose that $A(G)$ is feasible.
By \cref{lem:rigidity}, every task occupies one full residue class
and all those classes are disjoint.  Let $[r_v]_{nm_v}$ be the class occupied by the witness of vertex $v$, and define its
\emph{color} $c(v)$ as the integer in $\{0, 1, \dots , n-1\}$ such that 
\[
c(v) \equiv r_v \pmod n.
\]  
All marker primes exceed $n$, so $\gcd(n,m_v)=1$ for every $v$.  One basic property of marker products is
\begin{equation}\label{eq:marker-gcd}
  \gcd(m_u,m_v)>1
  \quad\Longleftrightarrow\quad
  \text{some triangle contains both $u$ and $v$}.
\end{equation}
Indeed, every triangle $T$ has its own prime $p_T$, and $p_T$ divides
$m_v$ exactly when $v\in T$.

\begin{lemma}[Every part uses every color]\label{lem:part-colors}
  The map $c$ is a bijection from each of $X,Y,Z$ to $\{0, 1, \dots , n-1\}$.
\end{lemma}

\begin{proof}
  Two distinct vertices $u,v$ in the same part lie in no common graph
  triangle.  Hence no marker prime divides both $m_u$ and $m_v$, so
  $\gcd(m_u,m_v)=1$ and
  \[
    \gcd(nm_u,nm_v)=n.
  \]
  If $c(u)=c(v)$, then
  $r_u\equiv r_v\pmod n$.
  By \cref{eq:crt-intersection}, the two witness classes would intersect.
  This would contradict dense rigidity.  Hence the $n$ vertices in each part
  receive $n$ distinct colors. 
\end{proof}

\begin{lemma}[A color determines a triangle]\label{lem:color-triangle}
  For each $t\in \{0, 1, \dots , n-1\}$, 
  let $x_t,y_t,z_t$ be the unique vertices of
  $X,Y,Z$, respectively, with color $t$.  Then
  $\{x_t,y_t,z_t\}$ forms a triangle.
\end{lemma}

\begin{proof}
  We first show that any two vertices of color $t$ share a graph triangle.
  Consider $x_t$ and $y_t$.  Suppose, for contradiction, that no graph
  triangle contains both of them.  By \cref{eq:marker-gcd}, $\gcd(m_{x_t},m_{y_t})=1$ and 
  \[
    \gcd(nm_{x_t},nm_{y_t})=n.
  \]
  By the definition of color,
  $r_{x_t}\equiv r_{y_t}\pmod n$.  
  Then \cref{eq:crt-intersection} implies that the two
  witness classes $[r_{x_t}]_{nm_{x_t}}$ and $[r_{y_t}]_{nm_{y_t}}$
  intersect, contradicting
  the pairwise disjointness of the residue classes in a feasible dense
  instance.  Hence some graph triangle contains both $x_t$ and $y_t$, and in
  particular $\{x_t,y_t\}\in E$.

  Applying the same argument to $\{x_t,z_t\}$ and $\{y_t,z_t\}$ gives the other
  two cross-edges.  Thus $\{x_t,y_t,z_t\}$ is a triangle by the definition of
  \TP{}.
\end{proof}

\begin{proposition}[Soundness]\label{prop:soundness}
  Let $G$ be a valid source graph in which every vertex belongs to a graph
  triangle.  If $A(G)$ is feasible, then $G$ is a YES-instance of \TP{}.
\end{proposition}

\begin{proof}
  For every $t\in \{0, 1, \dots , n-1\}$, 
  \cref{lem:color-triangle} produces a triangle
  $\{x_t,y_t,z_t\}$.  By \cref{lem:part-colors}, as $t$ varies these triangles
  contain every vertex of each part exactly once.  They are therefore a
  triangle partition of $G$.
\end{proof}

  Notice that this argument uses only witness tasks.  The cells are needed to
  build a packing in the completeness direction, but they cannot create a
  spurious source solution.

\section{Completeness}\label{sec:completeness}

\begin{proposition}[Completeness]\label{prop:completeness}
  If a valid source graph $G$ has a triangle partition, then $A(G)$ is
  feasible.
\end{proposition}
\begin{proof}
  By \cref{prop:density,lem:rigidity}, it suffices to assign one residue
  class to every task and prove that the assigned classes are pairwise
  disjoint.
  We verify this condition using the given triangle partition.
  Let $\mathcal P=\{T_t:t\in \{0, 1, \dots , n-1 \}\}$ be a triangle partition of $G$.

\paragraph{Reduction to local problem.}
  Recall from \cref{eq:block} that the block of a vertex $v$ consists of one
  witness task of period $nm_v$ and $m_v-3$ cell tasks of period $3nm_v$.
  Thus every task period has the form $nd$.  We call $d$ the
  \emph{local period} of the task.  In particular, the witness of $v$ has
  local period $m_v$, whereas each cell of $v$ has local period $3m_v$.

  Suppose that $v$ belongs to $T_t$.  For a task in the block of $v$ with
  local period $d$, choosing a \emph{local phase}
  $\gamma\in\{0,\ldots,d-1\}$
  means assigning that task the global residue class
  $[t+n\gamma]_{nd}$. We say that its \emph{local class} is $[t]_d$. 
  Thus the witness of $v$ will receive a local phase
  \[
    \alpha_v\in\{0,\ldots,m_v-1\}
  \]
  and the cell tasks will receive local phases
  \[
    \beta_{v,1},\ldots,\beta_{v,m_v-3}
      \in\{0,\ldots,3m_v-1\}.
  \]
  Their global classes are therefore
  \[
    [t+n\alpha_v]_{nm_v}
  \]
  for the witness and
  \[
    [t+n\beta_{v,\ell}]_{3nm_v}
    \qquad (1\le\ell\le m_v-3)
  \]
  for the cells.  We specify these local phases below.

  This form of the assignment already separates tasks whose vertices belong
  to distinct triangles of $\mathcal P$.  Indeed, suppose that tasks of period $dn$ and $en$ are
  associated with $T_s$ and $T_t$, respectively, where
  $s\ne t$.  Note that we chose all task periods to be divisible by $n$.  Their phases are
  congruent to $s$ and $t$, respectively, modulo $n$, and therefore cannot
  be congruent modulo $\gcd(dn, en)$.  By
  \cref{eq:crt-intersection}, the two residue classes are disjoint.

  It remains to separate tasks associated with the same partition triangle
  $T_t$.  Consider two such tasks with local periods $d,e$ and local phases
  $\alpha,\beta$.  Their global periods are $nd,ne$, and their global phases
  are $t+n\alpha,t+n\beta$.  Since $\gcd(nd,ne)=n\gcd(d,e)$,
  we have
  \begin{equation}\label{eq:cancel-outer-residue}
    t+n\alpha\equiv t+n\beta
      \pmod{\gcd(nd,ne)}
    \quad\Longleftrightarrow\quad
    \alpha\equiv\beta\pmod{\gcd(d,e)}.
  \end{equation}
  Consequently, for a fixed $T_t$, it suffices to choose the local phases so
  that any two distinct local classes do not intersect.

\paragraph{Assignment in the local problem.}
  Fix $t\in \{0, 1, \dots , n-1\}$ 
  and write
  \[
    T:=T_t=\{x,y,z\}
    \qquad (x\in X,\ y\in Y,\ z\in Z).
  \]
  Let $p :=p_T$ and let 
  \[
    Q_T:= \{0, 1, 2 \} \times \{0, 1, \dots p-1\}, 
  \]
   which can be represented as the $3\times p$ grid. 
   Define the local labels $c_x=0$, $c_y=1$, and $c_z=2$.
   For each $i\in\{x,y,z\}$, define
   \begin{align*}
     R_i:=
       &\{(q, r)\in Q_T:r = c_i\} \\
       \,\cup\,
      &\{(q, r)\in Q_T:r \ge 3,\
                      q = c_i \}.
  \end{align*}
  Thus $R_i$ consists of the full column $c_i$, together with row $c_i$ in
  every column outside $\{0,1,2\}$.  
  See Figure~\ref{fig:local-grid}.
  One sees that $|R_i|=3+(p-3)=p$ for each $i$, and $Q_T=R_x\mathbin{\dot\cup}R_y\mathbin{\dot\cup}R_z$.

  \begin{figure}
    \centering
    \pwLocalQuotientDiagram
    \caption{The region $R_i$ consists of the full column $c_i$ and row
      $c_i$ in every remaining column.}
    \label{fig:local-grid}
  \end{figure}

  We now choose the local phases for the block of a fixed vertex
  $i\in\{x,y,z\}$.  We define a map 
  \[ 
  \pi_i : \{0, 1, \dots , 3m_i-1\} \to Q_T
  \]
  by $\pi(s) = (q, r)$, where $s \equiv q \pmod 3$ and $s \equiv r \pmod p$. 
  Let $B_i:=\pi_i^{-1}(R_i)$. 
  Since $3p | 3m_i$, every element of $Q_T$ has exactly $\frac{3m_i}{3p}$
  preimages under $\pi_i$.  Therefore, 
  \[
  |B_i|=\frac{3m_i}{3p} |R_i|=m_i.
  \]

  Assign the witness of $i$ the local phase $\alpha_i:=c_i$, that is, the local residue class is $[c_i]_{m_i}$. 
  This class splits as 
  \begin{equation}\label{eq:witness-split}
    [c_i]_{m_i}
      = [c_i]_{3m_i} \mathbin{\dot\cup} [c_i + m_i]_{3m_i} \mathbin{\dot\cup} [c_i + 2 m_i]_{3m_i}.
  \end{equation}
  Let $W_i := \{c_i, c_i + m_i, c_i + 2m_i\}$. 
  Since $p\mid m_i$, all of $\pi(c_i)$, $\pi(c_i + m_i)$, and $\pi(c_i + 2m_i)$ lie in column $c_i$ of $Q_T$, 
  which implies that   $W_i \subseteq B_i$.
  The set $B_i\setminus W_i$ therefore has exactly $m_i-3$ elements.
  Choose the local cell phases bijectively so that
  \[
    \{\beta_{i,1},\ldots,\beta_{i,m_i-3}\}
      =B_i\setminus W_i.
  \]
  Performing this construction for each $i\in\{x,y,z\}$ completes the
  assignment of all task classes associated with $T_t$.

\paragraph{Validity of the assignment.}

  We verify that these local classes are pairwise disjoint.

  First consider two tasks in the same vertex block $i$.  
  Since two distinct cell tasks are assigned distinct integers from $B_i \subseteq \{0, 1, \dots , 3m_i-1\}$, we obtain
   $\beta_{i,\ell}\not\equiv\beta_{i,\ell'} \pmod{3m_i}$
    for $\ell\ne\ell'$, so the classes $[\beta_{i,\ell}]_{3m_i}$ and $[\beta_{i, \ell'}]_{3m_i}$ are disjoint.

  Next compare the witness of $i$ with a cell having local phase
  $\beta_{i,\ell}$.  The set $W_i$ consists precisely of the residues modulo
  $3m_i$ that are congruent to $c_i$ modulo $m_i$.  Since $\beta_{i,\ell}\notin W_i$, 
  we obtain $\beta_{i,\ell}\not\equiv c_i\pmod{m_i}$.
  Therefore, $[\beta_{i,\ell}]_{3m_i}$ and $[c_i]_{m_i}$ are disjoint.
  
  It remains to compare tasks belonging to two distinct vertex blocks $i,j\in\{x,y,z\}$.
  If both tasks are witnesses, their local residue classes are $[c_i]_{m_i}$ and $[c_j]_{m_j}$. 
  Since $p\mid\gcd(m_i,m_j)$ and $p>3$, we obtain 
  \[
  c_i\not\equiv c_j\pmod{\gcd(m_i,m_j)}. 
  \]
  Therefore, $[c_i]_{m_i}$ and $[c_j]_{m_j}$ are disjoint.
  
  Suppose next that one task is the witness of $i$ and the other is a
  cell of $j$ with local phase $\beta_{j,\ell}$.  Since $\pi_j(\beta_{j,\ell})\in R_j$,
  the column containing $\pi_j(\beta_{j,\ell})$ is either $c_j$ or lies outside
  $\{0,1,2\}$.  In either case, it is different from $c_i$.  Thus $\beta_{j,\ell}\not\equiv c_i\pmod p$. 
  Since $p\mid\gcd(m_i,3m_j)$,
  it follows that
  \[
    \beta_{j,\ell}\not\equiv c_i
      \pmod{\gcd(m_i,3m_j)}. 
  \]
  Therefore, $[c_i]_{m_i}$ and $[\beta_{j,\ell}]_{3m_j}$ are disjoint.

  Finally, suppose both tasks are cells, with local phases
  $\beta_{i,\ell}$ and $\beta_{j,\ell'}$.  
  Then, $\pi_i(\beta_{i,\ell}) \in R_i$ and $\pi_j(\beta_{j,\ell'}) \in R_j$.  
  Since $R_i$ and $R_j$ are disjoint, $\beta_{i,\ell}\not\equiv\beta_{j,\ell'}\pmod{3p}$.
  Since $3p\mid\gcd(3m_i,3m_j)$,
  we obtain
  \[
    \beta_{i,\ell}\not\equiv\beta_{j,\ell'}
      \pmod{\gcd(3m_i,3m_j)}.
  \]
  Therefore, $[\beta_{i,\ell}]_{3m_i}$ and $[\beta_{j,\ell'}]_{3m_j}$ are disjoint.

  These cases cover every pair of distinct tasks associated with $T_t$.
  By \cref{eq:cancel-outer-residue}, since their local phases do not overlap, their global phases also do not overlap.  Since $T_t$ was
  arbitrary, this holds within every triangle of $\mathcal P$.  Tasks
  associated with distinct triangles of $\mathcal P$ were already separated
  by their different residues modulo $n$.  Hence all assigned residue
  classes are pairwise disjoint.

  Since $A(G)$ has density one by \cref{prop:density},
  \cref{lem:rigidity} implies that these classes form an exact cover of
  $\Z$ and that $A(G)$ is feasible.
\end{proof}

\section{Explicit size and proof of the main theorem}\label{sec:size}

As noted in the source-problem definition, every valid source graph has at
most $3n$ triangles, and hence $|\mathcal T|\le3n$.

Let $\pi(x)$ denote the number of primes at most $x$, and let $p_j$ denote
the $j$th prime.  Choose the first $|\mathcal T|$ primes larger than
$\max\{3,n\}$.  Their indices in the prime sequence are at most
$\pi(\max\{3,n\})+3n\le6n$.  The standard bound $p_j=O(j\log j)$,
for example from Rosser and Schoenfeld~\cite{RosserSchoenfeld1962}, gives
\begin{equation}\label{eq:prime-bound}
  p_T=O(n\log n).
\end{equation}
Each $m_v$ is the product of the marker primes for the triangles containing
$v$.  There are at most three such primes, so \cref{eq:prime-bound} yields
\[
  m_v=O(n^3\log^3 n).
\]
The reduction explicitly lists $m_v-2$ tasks for $v$: one witness and
$m_v-3$ cells.  Therefore the total number of target tasks is
\begin{equation}\label{eq:task-bound}
  k=\sum_{v\in V(G)}\bigl(1+(m_v-3)\bigr)
   =O(n^4\log^3 n),
\end{equation}
and the largest period is
\[
  3nm_v=O(n^4\log^3 n).
\]
In particular, every period is polynomially bounded in $n$.  These bounds also
show that the explicit task list $A(G)$ can be generated in polynomial time.

\begin{proof}[Proof of \cref{thm:main}]
  \TP{} is NP-hard by \cref{thm:source}.  Given a source instance $G$, the
  construction of $A(G)$ is polynomial-time and has density exactly one by
  \cref{prop:density}.  Moreover,
  \cref{prop:soundness,prop:completeness} show that
  \[
    G\text{ has a triangle partition}
    \quad\Longleftrightarrow\quad
    A(G)\text{ is feasible}.
  \]
  Thus \DPP{} is NP-hard.
  By \cref{eq:task-bound} and the bound on the largest period,
  \[
    \sum_{j=1}^k a_j
      \le k\max_j a_j
      =O(n^8\log^6 n).
  \]
  Hence the output length remains polynomial even when all periods are encoded
  in unary, so the reduction proves NP-hardness for the unary representation.
  Together with the NP-membership above, this proves the theorem.
\end{proof}

\section*{Acknowledgments and use of AI}
The proof was initially found by GPT-5.5. The authors subsequently checked and revised it and take full responsibility for its correctness.
Bingkai Lin thanks Chao Xu for sharing his experience using Codex for
mathematical research in the blog post
\href{https://chaoxu.prof/posts/2026-07-18-ai-agents-for-the-working-mathematician.html}
{\emph{AI Agents for the Working Mathematician}}.

\newpage
\appendix
\section{NP-completeness of Sparse Tripartite Triangle Partition}\label{app:source}

For completeness, we prove \cref{thm:source}. The reduction uses the triangle-partition gadget of Dyer and Frieze~\cite{DyerFrieze1986}.

\begin{definition}[3-Dimensional Matching]\label{def:3dm}
  An instance consists of three pairwise disjoint sets $R,B,Y$ of equal size
  and a set of triples
  \[
    \mathcal S\subseteq R\times B\times Y.
  \]
  The question is whether there exists a collection of pairwise disjoint
  triples from $\mathcal S$ that covers
  $R\mathbin{\dot\cup}B\mathbin{\dot\cup}Y$.

  For an element $e\in R\cup B\cup Y$, define $
    \operatorname{occ}(e)
      :=|\{\tau\in\mathcal S:e\in\tau\}|$.
  We write $(2,3)$-\textsc{3-Dimensional Matching} for the restriction in
  which $\operatorname{occ}(e)\in\{2,3\}$ for every $e\in R\cup B\cup Y.$
\end{definition}

\begin{theorem}[\mbox{Dyer--Frieze~\cite[Theorem 2.3]{DyerFrieze1986}}]\label{thm:bounded-3dm}
  $(2,3)$-\textsc{3-Dimensional Matching} is NP-complete.
\end{theorem}

Dyer and Frieze prove the stronger statement with an additional planarity
restriction.  We use only the bounded-occurrence consequence stated above.

\begin{proof}[Proof of \cref{thm:source}]
  Membership of \TP{} in $\mathrm{NP}$ is immediate: for a valid graph with
  $n$ vertices in each tripartition class, a list of $n$ vertex-disjoint
  triangles covering all vertices is a polynomially verifiable certificate.

  For NP-hardness, let $\mathcal I=(R,B,Y,\mathcal S)$
  be an instance of $(2,3)$-\textsc{3-Dimensional Matching}, where
  \[
    R=\{r_1,\ldots,r_q\},\qquad
    B=\{b_1,\ldots,b_q\},\qquad
    Y=\{y_1,\ldots,y_q\}.
  \]
  We construct a graph $G_{\mathcal I}$ whose three tripartition classes
  correspond to $R,B,Y$.

  Create one \emph{element vertex} for every member of
  $R\cup B\cup Y$, placing it in the corresponding tripartition class.
  For every source triple
  \[
    \tau=(r,b,y)\in\mathcal S,
  \]
  create three \emph{apex vertices} $a_{\tau,r}$, $a_{\tau,b}$, and $a_{\tau,y}$,
  placing each apex in the same tripartition class as its corresponding
  element.  Join the three apex vertices in a central triangle.

  \begin{figure}[H]
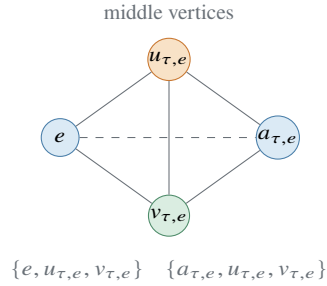

    \centering
    \pwDiamondDiagram
    \caption{One diamond.  The dashed apex--element pair indicates the
      missing edge.  The two displayed triples are the element-side and
      apex-side triangles.}
    \label{fig:diamond}
  \end{figure}

  For every pair $(\tau,e)$ with $e\in\tau$, create two
  \emph{middle vertices} $u_{\tau,e}$ and $v_{\tau,e}$, one in each of the
  two tripartition classes different from that of $e$.  Join
  $u_{\tau,e}$ and $v_{\tau,e}$, and join each of them to both $e$ and
  $a_{\tau,e}$.  These four vertices therefore induce a $K_4$ minus the
  edge $e a_{\tau,e}$; we call this subgraph a \emph{diamond}.

  The diamond contains exactly the two triangles
  \[
    \{e,u_{\tau,e},v_{\tau,e}\}
    \quad\text{and}\quad
    \{a_{\tau,e},u_{\tau,e},v_{\tau,e}\},
  \]
  called respectively its \emph{element-side triangle} and
  \emph{apex-side triangle}.  The diamond and its two possible covering
  triangles are shown in \cref{fig:diamond}.

  The \emph{triple gadget} associated with
  $\tau=(r,b,y)$ consists of its central apex triangle and the three diamonds
  corresponding to $r,b,y$.  Element vertices are shared by the gadgets of
  the source triples containing them, whereas all apex and middle vertices
  are private to their respective gadgets.  There are no other vertices or
  edges.  The resulting gadget is shown in
  \cref{fig:triple-gadget}.

  \begin{figure}
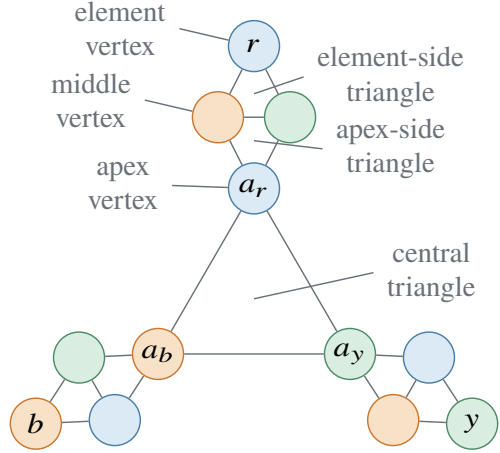

    \centering
    \scalebox{1.35}{\pwTripleGadgetDiagram}
    \caption{The triple gadget for $\tau=(r,b,y)$.  Colors indicate the three
      tripartition classes; the $\tau$ subscripts and middle-vertex labels
      are suppressed.}
    \label{fig:triple-gadget}
  \end{figure}

  We first verify that $G_{\mathcal I}$ is a valid instance of \TP{}.
  The graph is simple and tripartite by construction.  Each middle vertex
  has degree three, each apex vertex has degree four, and an element vertex
  $e$ has degree $2\operatorname{occ}(e)\in\{4,6\}$.
  Hence the maximum degree is at most six.

  The graph is also balanced.  Each source triple contributes one apex and
  two middle vertices to each tripartition class.  Consequently, every class
  contains $n':=q+3|\mathcal S|$
  vertices.

  We next identify all triangles of $G_{\mathcal I}$.  Fix a middle vertex,
  say $u_{\tau,e}$.  Its only neighbors are $v_{\tau,e}$, $e$, and $a_{\tau,e}$.
  The vertex $v_{\tau,e}$ is adjacent to both $e$ and $a_{\tau,e}$, whereas
  $e$ and $a_{\tau,e}$ are not adjacent.  Therefore, the only triangles
  containing $u_{\tau,e}$, and likewise the only triangles containing
  $v_{\tau,e}$, are the element-side and apex-side triangles of their
  diamond.

  A triangle containing no middle vertex cannot contain an element vertex,
  because element vertices are adjacent only to middle vertices.  The only
  remaining triangles are therefore the central triangles, since
  apex--apex edges occur only within the central triangle of a single triple
  gadget.  Thus the complete list of graph triangles consists of the two
  triangles in each diamond and the central triangle in each triple gadget.
  In particular, no graph triangle contains internal vertices from two
  distinct triple gadgets.

  It follows that every middle vertex belongs to exactly two graph triangles,
  every apex vertex belongs to exactly two graph triangles, and every element
  vertex $e$ belongs to exactly
  $\operatorname{occ}(e)\in\{2,3\}$ graph triangles.  Hence every vertex
  belongs to at least one and at most three triangles, as required in the
  definition of \TP{}.

  The construction contains $3q+9|\mathcal S|$
  vertices and $18|\mathcal S|$ edges, and can therefore be carried out in
  polynomial time.

  It remains to characterize the triangle partitions of
  $G_{\mathcal I}$.  As is also visible in \cref{fig:diamond}, the two middle
  vertices of every diamond must be covered together by exactly one of its
  element-side and apex-side triangles.  The interaction of the three
  diamonds through the central triangle, shown in
  \cref{fig:triple-gadget}, leaves exactly two possible modes.

  If the central apex triangle is used, then all three apex vertices are
  already covered.  None of the three apex-side triangles can be used, so
  all three diamonds must use their element-side triangles.  If the central
  triangle is not used, every apex must instead be covered by the apex-side
  triangle of its diamond.  Thus every triple gadget has exactly two modes:
  \begin{itemize}
    \item the \emph{selected mode}, consisting of the central triangle and
      the three element-side triangles; this mode covers the three source
      element vertices of the triple;
    \item the \emph{unselected mode}, consisting of the three apex-side
      triangles; this mode covers no source element vertex.
  \end{itemize}

  Suppose first that $G_{\mathcal I}$ has a triangle partition, and let
  \[
    \mathcal M
      :=\{\tau\in\mathcal S:
           \text{the gadget of $\tau$ uses the selected mode}\}.
  \]
  The only triangles containing an element vertex $e$ are the element-side
  triangles in the gadgets of the source triples containing $e$.  Since the
  triangle partition covers $e$ exactly once, exactly one triple in
  $\mathcal M$ contains $e$.  Hence the triples in $\mathcal M$ are pairwise
  disjoint and cover every element of
  $R\mathbin{\dot\cup}B\mathbin{\dot\cup}Y$.  Thus $\mathcal M$ is a
  perfect 3-dimensional matching.

  Conversely, suppose that the source instance has a perfect
  3-dimensional matching $\mathcal M\subseteq\mathcal S$.  Use the selected
  mode for every gadget corresponding to a triple in $\mathcal M$, and use
  the unselected mode for every other gadget.  These choices cover every
  apex and middle vertex exactly once.  Since $\mathcal M$ covers every
  source element exactly once, they also cover every element vertex exactly
  once.  The chosen graph triangles therefore form a triangle partition of
  $G_{\mathcal I}$.

  We have proved that
  \[
    \mathcal I\text{ has a perfect 3-dimensional matching}
    \quad\Longleftrightarrow\quad
    G_{\mathcal I}\text{ has a triangle partition}.
  \]
  The reduction is polynomial, and its output is always a valid instance of
  \TP{}.  Therefore \TP{} is NP-hard.  Together with membership in
  $\mathrm{NP}$, this proves that \TP{} is NP-complete.
\end{proof}

\bibliographystyle{alphaurl}
\bibliography{references}

\end{document}